\documentclass[conference]{IEEEtran}
\usepackage{cite}
\usepackage{amsmath,amssymb,amsfonts}
\usepackage{amsthm}
\usepackage{graphicx}
\usepackage{booktabs}
\usepackage{tikz}
\usetikzlibrary{positioning,arrows.meta,calc,shapes.geometric}
\usepackage{makecell}
\usepackage{newtxtext}
\usepackage{url}
\usepackage[hidelinks]{hyperref}
\usepackage{tabularx}
\usepackage{subcaption}
\usepackage{placeins}
\usepackage{algorithm}
\usepackage{algpseudocode}

\newtheorem{definition}{Definition}
\newtheorem{proposition}{Proposition}

\newtheorem{corollary}{Corollary}

\def\BibTeX{{\rm B\kern-.05em{\sc i\kern-.025em b}\kern-.08em
    T\kern-.1667em\lower.7ex\hbox{E}\kern-.125emX}}

\begin{document}

\title{Quantum Workload Privacy Beyond Data Confidentiality}

\author{
    \IEEEauthorblockN{Shaunak Suresh Pawar, Samuel Punch, and Krishnendu Guha}
    \IEEEauthorblockA{
        School of Computer Science and Information Technology,
        University College Cork, Ireland\\
        Email: {pawar2003.shaunak@gmail.com,samuel.punch, k.guha}@ucc.ie
    }
}

\maketitle
\thispagestyle{plain}\pagestyle{plain}   

\begin{abstract}
Quantum computing increasingly relies on remote execution platforms, requiring users to submit computational workloads to third-party quantum hardware. While existing privacy mechanisms primarily protect quantum states, input data, and computational outputs, they provide limited protection for information encoded indirectly in the structure of a submitted workload. This work investigates a previously underexplored confidentiality channel in which hardware-aware compilation exposes information about the scientific computation being performed. We introduce the concept of \emph{computational-intent leakage}, where compilation and execution artefacts provide observable signatures that can be correlated with hidden characteristics of a scientific workload. The underlying mechanism arises from the interaction between logical circuit structure and the connectivity constraints of quantum processors. Routing, gate decomposition, circuit depth, and SWAP insertion are not independent of the logical workload; instead, they can vary systematically with parameters that define the underlying scientific computation. Consequently, an observer with access to compilation or execution metadata may distinguish workloads that differ in modelling structure, discretisation, or geometry without observing the protected quantum state or final computational result. We formulate this threat using Scientific-Intent Indistinguishability ($\mathrm{SCI\text{-}IND}$) and analyse the conditions under which deterministic hardware-aware compilation produces distinguishable workload signatures. Experimental evaluation on the 156-qubit IBM Heron processor \texttt{ibm\_fez} demonstrates substantial separability between structurally distinct workloads, including PDE discretisations with different cycle structures and molecular workloads with different geometries. We further examine near-isomorphic workloads, showing that routing-scaling behaviour can reveal hidden discretisation characteristics even when direct structural classification becomes difficult. Across the evaluated cases, conventional gate-padding provides no meaningful confidentiality improvement and can instead introduce substantial fidelity degradation. These results demonstrate that protecting quantum data alone does not guarantee workload confidentiality and motivate compilation-aware mechanisms for concealing computational intent in delegated quantum computing.
\end{abstract}


\section{Introduction}
\label{sec:intro}

Delegated quantum scientific computing introduces a confidentiality challenge that standard quantum privacy models overlook. In high-value scientific and engineering applications, the primary sensitive asset is not the numerical output but the modelling structure of the workload, the physical regime, algorithmic configuration, and discretisation choices that define a proprietary simulation methodology. We term this Scientific Intellectual Property (SIP). Quantum algorithms for partial differential equations (PDEs) instantiate this risk most sharply, with impact across fluid dynamics, electromagnetics, and quantitative finance~\cite{Childs_2021,Kyriienko_2021}, because the link between PDE modelling choices and circuit topology is unusually direct. The threat is structural: it extends to any delegated workload whose hidden intent shapes operator connectivity, including VQE-based quantum chemistry. In practice, such workloads execute on multi-tenant cloud platforms, transferring computation across a trust boundary to a provider that manages quantum hardware and the full compilation stack.

Existing quantum privacy research protects the \emph{computational payload}: blind quantum computation (BQC)~\cite{Broadbent_2009,fitzsimons2016privatequantumcomputationintroduction} hides what the computation \emph{is}, and timing side-channel work~\cite{lu2024quantumleaktimingsidechannel,dong2025exploitingtimingsidechannelsquantum} addresses when it runs. Neither addresses the confidentiality of the \emph{scientific structure} encoded in the workload itself. We show that SIP is physically coupled to the hardware-compilation artefacts the provider observes as a routine byproduct of managing the compilation stack. As illustrated in Fig.~\ref{fig:system_model}, the delegation workflow exposes a rich observational surface precisely where the trust boundary lies.

\medskip
\noindent\textbf{Contributions.}
\begin{itemize}
    \item \textbf{SCI-IND.} A game-based security notion for scientific intent indistinguishability, with an analytical result (Proposition~\ref{prop:infeasibility}, Corollary~\ref{cor:asymptotic}) showing that passive SCI-IND security is asymptotically unachievable under any routing-optimal compiler on constrained hardware.

    \item \textbf{Physics-to-artefact mapping.} Three deterministic coupling modes (topological, complexity, operational) linking hidden scientific attributes to provider-visible execution signatures, grounded in graph-theoretic embeddability.

    \item \textbf{Algorithm-agnostic scale recovery.} Under a strict family-holdout protocol, a passive observer extracts the spatial discretisation $N$ from a routing-overhead scaling law $\Delta_{\mathrm{routing}} \sim N^k$, with stable exponents across solver families ($k = 1.673$ versus $k = 1.836$). Leakage persists even when the underlying connectivity graphs are near-isomorphic, the regime not covered by Proposition~\ref{prop:infeasibility}, and the recovered signal generalises across algorithmic implementations.

    \item \textbf{Testbed evaluation.} Macro-F1 up to $0.93$ for boundary topology, discretisation scale, hardware stability under calibration drift, and accuracy targets, on a hardware-aware compilation testbed calibrated to near-term quantum constraints.

    \item \textbf{Hardware confirmation on IBM Heron r2.} Maximal-separation baseline runs on \texttt{ibm\_fez} (156 qubits) confirm the predicted $C_8$ versus $C_3$ routing gap on a girth-12 Heavy-Hex lattice (Macro-F1 $= 1.000 \pm 0.000$, $d_z = 7.701$).

    \item \textbf{Cross-domain transfer to VQE chemistry.} The same routing-gap signature reproduces for hidden molecular geometry (linear H$_2$ versus bent H$_2$O, $d_z = 49.173$).

    \item \textbf{Defence analysis.} Passive gate-padding is topologically insufficient: no passive normalisation reduces adversarial advantage without catastrophic fidelity loss.
\end{itemize}

\section{Background and Related Work}
\label{sec:related}

Table~\ref{tab:comparison} positions this work relative to four relevant research areas. Quantum PDE algorithms establish that modelling choices directly govern qubit count and circuit depth~\cite{Harrow_2009,Childs_2021,Montanaro_2016,Costa_2019,Kyriienko_2021,Brearley_2024}. Quantum side-channel studies show that timing and resource traces fingerprint circuits~\cite{lu2024quantumleaktimingsidechannel,dong2025exploitingtimingsidechannelsquantum}, but focus on generic circuit properties rather than scientific modelling intent. BQC protocols~\cite{Broadbent_2009,Gheorghiu_2015} hide circuit semantics from the server; their relationship to SCI-IND is discussed in Section~\ref{sec:sci_ind}. Hardware variability studies~\cite{dasgupta2020characterizingstabilitynisqdevices, tannu2019mitigating} characterise routing and calibration effects as performance artefacts, not confidentiality channels. Multi-programming and adversarial crosstalk work focuses on active user-vs-user threats~\cite{das2019case,saki2020analysis}. No prior work examines the passive, provider-vs-user confidentiality of scientific modelling intent through hardware-compilation artefacts.

\begin{table}
\centering
\renewcommand{\arraystretch}{1.15}
\caption{Positioning relative to prior work.}
\label{tab:comparison}
\setlength{\tabcolsep}{3pt}
\begin{tabular}{@{}lcccc@{}}
\toprule
\textbf{Research Area} &
\makecell{Exec.\\artefacts} &
\makecell{Sci.\\inference} &
\makecell{HW-coupled\\analysis} &
\makecell{PDE\\context} \\
\midrule
Quantum PDE algorithms       & --  & -- & --      & \checkmark \\
Timing/usage side-channels   & \checkmark & -- & partial & -- \\
HW variability studies       & partial & -- & \checkmark & -- \\
Blind / verifiable QC        & limited & -- & --     & -- \\
\midrule
\textbf{This work}           & \checkmark & \checkmark & \checkmark & \checkmark \\
\bottomrule
\end{tabular}
\end{table}

\section{System and Threat Model}
\label{sec:system}

\begin{figure*}[t]
    \centering
    \includegraphics[width=\linewidth]{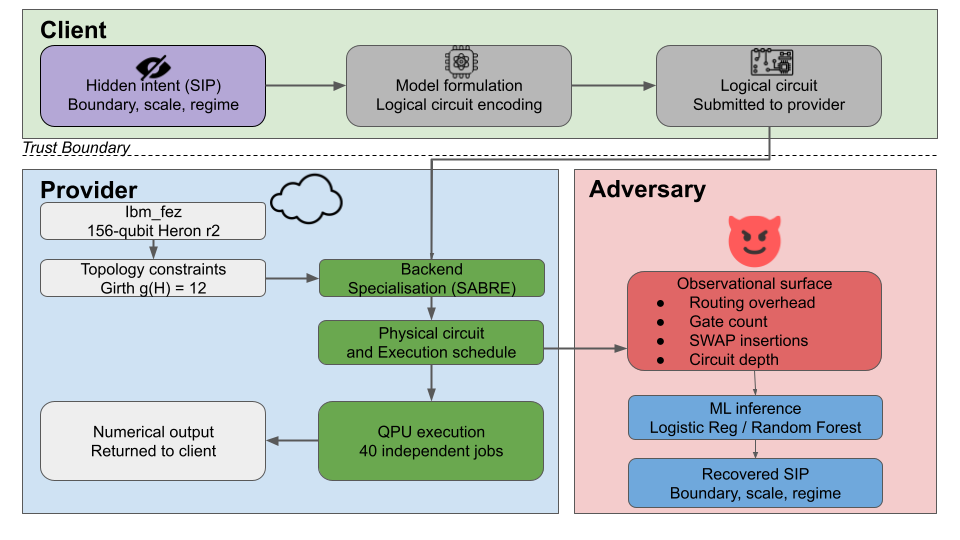} 
    \caption{Delegation pipeline and observational surface. (1) Client encodes hidden modelling intent (boundary topology, grid resolution) into a logical circuit. (2) The circuit is submitted to the cloud provider. (3) Provider's compiler maps the logical circuit onto the fixed hardware topology, inserting routing SWAPs that depend on the hidden connectivity. (4) The provider observes execution artefacts (routing overhead, gate composition, depth) that reveal structural properties of the original scientific problem.}
    \label{fig:system_model}
\end{figure*}

\subsection{Delegated Execution Model}
Fig.~\ref{fig:system_model} illustrates the delegation workflow and the corresponding threat landscape. A \emph{scientific client} defines a PDE instance, encoding boundary topology, grid resolution, and solver regime into a logical circuit, and delegates execution to a \emph{cloud provider} that manages QPU compilation and scheduling. Information divides into three domains: (i)~\textbf{hidden modelling intent (SIP)}: boundary conditions, grid resolution, accuracy target, intended to remain strictly confidential; (ii)~\textbf{numerical output}: the final result, often eventually disclosed; (iii)~\textbf{execution artefacts}: routing overhead, gate composition, circuit depth, and transpilation latency, inherently visible to the provider as byproducts of managing the compilation stack.

\subsection{Adversary: The Conditioned Scientific Provider}
\label{subsec:adversary}

As depicted in the untrusted cloud zone of Fig.~\ref{fig:system_model}, the adversary is an honest-but-curious cloud provider operating a Scientific-Computing-as-a-Service (SCaaS) platform. Because the provider supplies PDE-solving API templates, it possesses \emph{domain conditioning}: knowledge of the solver family and baseline execution footprints against which deviations in compiled artefacts are detectable. The adversary's visibility is restricted to the information inherently exposed by delegating the compilation stack: physical qubit mapping, routed circuit depth, gate-mixture statistics, and transpilation latency. The adversary is purely \emph{passive}, no circuit modification, noise injection, or active probing, and relies entirely on the physical-to-compilation coupling: changes in the physics $\mathcal{P}$ force measurable changes in backend specialisation, which map telemetry back to scientific intent. This is the weakest plausible adversary; stronger models constitute strictly greater threats not evaluated here.

f
The leakage arises because the compiler must route logical qubits onto a fixed hardware connectivity graph. A PDE solver with periodic boundary conditions creates a ring of interactions that cannot be natively embedded on a chip whose shortest cycle is large (e.g., Heavy-Hex girth 12). The compiler compensates by inserting SWAP gates, which increase circuit depth and gate counts. The amount of extra routing depends directly on the hidden boundary choice, and this routing cost is visible to the provider as an execution artefact. In essence, the hardware-constrained compilation process \emph{imprints} the problem's connectivity structure onto measurable, provider-visible metrics.

This principle is not confined to PDE workloads; any delegated quantum program whose hidden parameters modify logical connectivity is vulnerable. PDE and VQE are concrete instances where the mapping from physical structure to connectivity is especially direct, making leakage highly interpretable. The phenomenon resembles classical side-channels (e.g., cache timing), but with a critical difference: classical platforms have mature mitigations (ORAM, constant-time compilation) with well-understood cost tradeoffs. Our defence evaluation (Section~\ref{sec:defence}) shows that gate-padding, the direct quantum analogue of constant-time compilation, fails because the leak is topological rather than merely volumetric.

\section{Scientific-Intent Indistinguishability (SCI-IND)}
\label{sec:sci_ind}

\subsection{Workload Model}
A delegated workload is a tuple $W = (x, s)$ where $x \in \mathcal{X}$ is the \emph{public profile} (solver family, nominal interaction budget, hardware target) and $s \in \mathcal{S}$ is the \emph{hidden scientific intent} (SIP). The compilation function maps workload, hardware description $H \in \mathcal{H}$, and randomness $\rho \in \mathcal{R}$ to a physical circuit and schedule:
\[
    (C_{\mathrm{phys}}, \sigma) \;\leftarrow\; \mathsf{Comp}(W, H, \rho)
\]
The provider observes only the execution transcript $\tau \;\leftarrow\; \mathsf{Trace}(W, H, \rho) = \pi(\mathsf{Comp}(W, H, \rho))$, where $\pi$ extracts routed depth, two-qubit overhead, swap fraction, CX fraction, depth overhead ratio, and transpilation latency.

\subsection{Relation to Existing Notions}
BQC hides what the computation \emph{is}; SCI-IND addresses whether the \emph{operational signature} reveals hidden scientific structure. A scheme achieving BQC circuit privacy may still fail SCI-IND if its compilation footprint varies systematically with the client's modelling choices. Conversely, SCI-IND does not imply BQC circuit privacy. The two notions are orthogonal.

\subsection{The SCI-IND Security Game}
\begin{enumerate}
    \item Challenger samples $H \leftarrow \mathcal{H}$ and makes $H$ public.
    \item Adversary outputs $W_0 = (x, s_0)$, $W_1 = (x, s_1)$ with identical public profile $x$ but $s_0 \neq s_1$.
    \item Challenger samples $b \in_R \{0,1\}$, $\rho \leftarrow \mathcal{R}$, returns $\tau \leftarrow \mathsf{Trace}(W_b, H, \rho)$.
    \item Adversary outputs guess $b'$.
\end{enumerate}

\begin{definition}[SCI-IND Security]
Scheme $\Pi = (\mathsf{Comp}, \mathsf{Trace})$ is \emph{SCI-IND secure} if for all PPT adversaries $\mathcal{A}$:
$\mathrm{Adv}^{\mathrm{SCI\text{-}IND}}_{\mathcal{A}}(\lambda) = \left|\Pr[b' = b] - \tfrac{1}{2}\right| \leq \mathsf{negl}(\lambda)$.
\end{definition}

\subsection{Analytical Limits}
\label{subsec:hardness}

\begin{proposition}[Baseline Observational Lemma]
\label{prop:infeasibility}
Let $\mathsf{Comp}$ be a hardware-aware compiler that deterministically minimises routing overhead on fixed hardware $H$. For any $W_0, W_1$ whose hidden intents induce logically non-isomorphic connectivity graphs $G_0, G_1$ with an edge in their symmetric difference whose realisation cost differs under the routing-optimal embedding into $H$, and assuming compilation randomness $\rho$ induces normally distributed routing overhead, there exists an efficient distinguisher $\mathcal{A}$ with advantage
\begin{equation}
    \mathrm{Adv}^{\mathrm{SCI\text{-}IND}}_{\mathcal{A}} = \frac{1}{2}\mathrm{erf}\!\left(\frac{\Delta_{\mathrm{routing}}}{2\sqrt{2}\,\sigma_\rho}\right),
\end{equation}
where $\Delta_{\mathrm{routing}}$ is the expected difference in routing overhead and $\sigma_\rho$ the standard deviation induced by $\rho$.
\end{proposition}

\begin{proof}[Proof sketch]
Under deterministic routing minimization, $\mu_i = \mathbb{E}_\rho[\tau_{\mathrm{routing}}\mid W_i]$ are fixed and $|\mu_0 - \mu_1| = \Delta_{\mathrm{routing}} > 0$. With Gaussian routing noise, the likelihood ratio test reduces to thresholding $\tau_{\mathrm{routing}}$ at $\theta = (\mu_0+\mu_1)/2$, yielding the claimed error-function advantage. The distinguisher runs in $O(1)$.
\end{proof}

\begin{corollary}[Asymptotic Infeasibility]
\label{cor:asymptotic}
Whenever $\Delta_{\mathrm{routing}} / \sigma_\rho \geq c > 0$, $\mathrm{Adv}^{\mathrm{SCI\text{-}IND}}_{\mathcal{A}} \geq \frac{1}{2}\mathrm{erf}\bigl(\frac{c}{2\sqrt{2}}\bigr) > 0$. Achieving $\mathrm{Adv} \to 0$ requires $\sigma_\rho / \Delta_{\mathrm{routing}} \to \infty$, i.e., injecting noise that destroys computational utility.
\end{corollary}

Proposition~\ref{prop:infeasibility} covers topologically distinct graphs. Near-isomorphic graphs arising from continuous parameter variations (e.g., $N=14$ vs.\ $N=16$) are not analytically bounded by this result; our empirical evaluation (Phase~2) shows that a scaling law $\Delta_{\mathrm{routing}} \sim N^k$ still enables scale recovery, confirming the threat extends beyond strict non-isomorphism.

\section{Physics-to-Artefact Coupling}
\label{sec:coupling}

The mapping $\mathcal{M}: \mathcal{P} \xrightarrow{f,\mathcal{H}} \mathcal{A}$ operates through three deterministic modes.
\textbf{Topological coupling.} A periodic boundary introduces a wrap-around cycle of length $n$ that cannot be natively embedded on hardware with large girth, forcing a measurable increase in SWAPs, two-qubit burden, and CX overhead. Dirichlet boundaries, by contrast, induce only local edges and require fewer SWAPs. The difference constitutes a hardware-coupled routing gap $\Delta_{\mathrm{routing}} > 0$.
\textbf{Complexity coupling.} Grid resolution $N$ governs volumetric scaling; depth $\propto N \cdot n_{\mathrm{steps}}$ and gate volume $\propto N^\alpha$ create stable, resolution-dependent signatures in routed depth and gate count.
\textbf{Operational coupling.} Accuracy target $\varepsilon$ determines Trotter steps $n_{\mathrm{steps}} \propto (1/\varepsilon)^{1/2}$, directly inflating depth, while physical regime (diffusion-to-advection ratio) influences transpilation latency. Because these couplings are rooted in the functional requirements of the physics, the resulting artefacts form a stable fingerprint.

\section{Testbed Evaluation}
\label{sec:evaluation}

We exhaustively map the leakage surface across four dimensions (boundary topology, discretisation scale, stability under drift, accuracy requirements) using a calibrated testbed that employs Qiskit SABRE routing with randomised seeds, $\{\mathrm{rz, sx, x, cx}\}$ basis at optimisation level~1. All tests use \emph{matched workload pairs} (identical logical interaction budget, differing only in the hidden attribute). Leakage is assessed via mutual information (MI) against pair-preserving permutation-null baselines, paired Cohen's $d_z$, and grouped held-out classification.

\subsection{Phases 1--4: Boundary, Scale, Stability, Accuracy}
\textbf{Phase 1 (Boundary topology).} Changing only the hidden boundary regime (periodic vs.\ Dirichlet) induces significant leakage. On the \texttt{ladder\_2x4} topology, all eight compilation features exceed the null 95\% band ($p_{\mathrm{Holm}} < 0.016$); \texttt{swap\_equiv} yields 0.624 bits MI. Random Forest achieves Macro-F1 $= 0.925 \pm 0.021$. Leakage strength depends on the \emph{interaction} between hardware topology and workload topology, not simply hardware constraint.

\textbf{Phase 2 (Discretisation scale).} Grid resolution $N \in \{4,6,8,10,12,16\}$ is recovered through a routing scaling exponent $k$ (\texttt{extra\_twoq} $\sim N^k$) that transfers across solver families (Family A $k = 1.673$, Family B $k = 1.836$). A Routed Attacker trained on Family A achieves 96.3\% adjacent accuracy on held-out Family B, confirming algorithm-agnostic leakage. On the \texttt{gridish} topology where the local stencil maps natively, leakage collapses (MI $= 0.011$ bits), proving the signal arises from physical mismatch, not logical volume.

\textbf{Phase 3 (Hardware stability).} Under simulated calibration drift $d \in [0.0,1.0]$, boundary inference on \texttt{gridish} maintains F1 $= 0.98$ across all drift levels; on the more sensitive \texttt{line} topology, F1 degrades gracefully from 0.97 to 0.64 while adjacent accuracy for scale remains $\geq 82\%$.

\textbf{Phase 4 (Accuracy targets).} High-accuracy ($\varepsilon = 10^{-4}$) versus low-accuracy ($\varepsilon = 10^{-2}$) workloads produce routed depth $\approx 800$ vs.\ $\approx 100$ ($d_z = 1.296$), a volumetric signal that persists across constrained and relaxed hardware. Cross-family evaluation (time-evolution to optimisation) yields Macro-F1 $> 0.90$, demonstrating structural compilation features rather than solver-specific templates.

\subsection{Phase 5: VQE Molecular Structure}
We instantiate the SCI-IND game with 4-qubit VQE ansätze: linear H$_2$ (path graph) versus bent H$_2$O (star/triangle connectivity). Both share 4 qubits, VQE solver family, and 9 logical CX gates. On a Heavy-Hex 8-qubit coupling map ($n=40$ matched pairs), all six extracted features are significant ($p < 0.001$). The bent geometry forces 43.33 routed depth vs.\ 20.00 for linear ($d_z = 49.173$) and $+4.00$ extra CX gates. Logistic Regression achieves Macro-F1 $= 1.000 \pm 0.000$, maximum empirical SCI-IND advantage. Hidden molecular geometry, a core SIP asset in drug discovery and materials science, leaks through the same topological mechanism.

\begin{table}[t]
\centering
\caption{Testbed evaluation summary. All phases produce Macro-F1 significantly above chance.}
\label{tab:Testbed_summary}
\setlength{\tabcolsep}{4pt}
\resizebox{\columnwidth}{!}{%
\begin{tabular}{llcc}
\toprule
Phase & Target attribute & LR F1 & RF F1 \\
\midrule
1 & Boundary topology  & $0.879 \pm 0.023$ & $0.925 \pm 0.021$ \\
2 & Grid resolution    & $0.484$ (cross-family) & $0.963$ adj.acc \\
3 & Stability (drift)  & \multicolumn{2}{c}{F1 $\geq 0.64$ at max drift} \\
4 & Accuracy target & 0.943 $\pm$ 0.023 & 0.935 $\pm$ 0.021 \\
5 & Molecular geometry & $1.000 \pm 0.000$ & $1.000 \pm 0.000$ \\
\bottomrule
\end{tabular}%
}
\end{table}

\section{Real-Hardware Validation on IBM Heavy-Hex}
\label{sec:hardware_validation}

\subsection{Experimental Configuration}
\textbf{Target device.} \texttt{ibm\_fez}: 156-qubit IBM Heron~r2, Heavy-Hex connectivity, girth $g(H) = 12$, native basis $\{\mathrm{rz, sx, x, cz}\}$.

\textbf{Topological basis.} Workload pairs on $n = 8$ logical qubits enforce the intra-family matching constraint under two boundary conditions: \textbf{Periodic ($C_8$)}: global wrap-around edge forms an 8-qubit ring; \textbf{Dirichlet ($C_3$)}: wrap-around replaced by a next-nearest-neighbour cross-link forming a local triangle. The Heavy-Hex lattice has girth 12, containing no cycle shorter than 12 qubits, so neither $C_8$ nor $C_3$ can be natively embedded without SWAP insertion. The $C_8$ ring requires a strictly longer minimum routing path, establishing a deterministic routing gap $\Delta_{\mathrm{routing}} > 0$ as predicted by Proposition~\ref{prop:infeasibility}. Fig.~\ref{fig:topology} illustrates the mechanism.

\begin{figure}[t]
    \centering
    \resizebox{\columnwidth}{!}{%
    \begin{tikzpicture}[
        node distance=0.55cm and 0.55cm,
        qnode/.style={circle, draw, fill=blue!15, minimum size=0.42cm, inner sep=0pt, font=\scriptsize},
        hnode/.style={circle, draw, fill=orange!20, minimum size=0.38cm, inner sep=0pt, font=\tiny},
        swap/.style={->, thick, red, dashed},
        edge/.style={thick, blue!60},
        hedge/.style={thick, orange!70}
    ]
    \node[font=\scriptsize\bfseries] at (0, 1.55) {Periodic: $C_8$ ring};
    \foreach \i in {0,...,7} {
        \node[qnode] (p\i) at ({360/8 * \i + 90}:1.1cm) {\i};
    }
    \foreach \i [evaluate=\i as \j using {int(mod(\i+1,8))}] in {0,...,7} {
        \draw[edge] (p\i) -- (p\j);
    }
    \node[font=\tiny, red] at (0, -1.45) {wrap-around $\Rightarrow$ long SWAP path};
    \begin{scope}[xshift=3.6cm]
    \node[font=\scriptsize\bfseries] at (0, 1.55) {Dirichlet: $C_3$ cross-link};
    \foreach \i in {0,...,7} {
        \node[qnode] (d\i) at ({-1.75 + \i*0.5}, 0) {\i};
    }
    \foreach \i [evaluate=\i as \j using {int(\i+1)}] in {0,...,6} {
        \draw[edge] (d\i) -- (d\j);
    }
    \draw[edge, green!60!black, very thick] (d2) to[bend left=35] (d4);
    \node[font=\tiny, green!50!black] at (0, -0.75) {local NNN cross-link};
    \node[font=\tiny, blue!70] at (0, -1.45) {short SWAP path};
    \end{scope}
    \node[font=\tiny, align=center] at (1.8, -2.15)
        {Heavy-Hex girth $= 12$: neither $C_8$ nor $C_3$ embeds natively\\
         $\Rightarrow$ compiler forced to insert SWAPs, but $\Delta_{\mathrm{routing}} > 0$};
    \end{tikzpicture}%
    }
    \caption{Leakage basis on IBM Heavy-Hex ($g(H)=12$).}
    \label{fig:topology}
\end{figure}
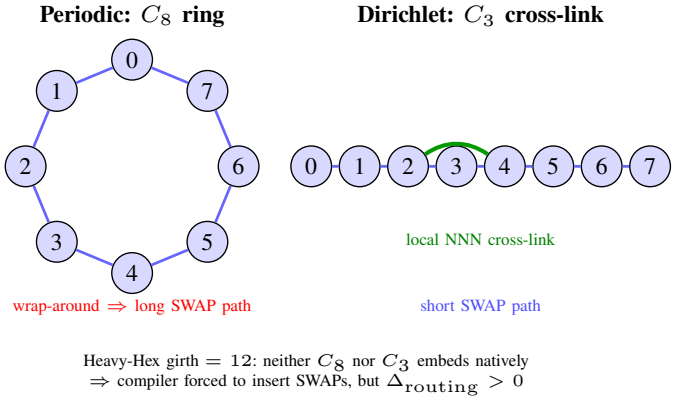

\textbf{Protocol.} Stage~1: $N_{\mathrm{pairs}} = 120$ matched pairs transpiled against the \texttt{ibm\_fez} coupling map, 5 seeds, optimisation level~3 (1200 circuits). Stage~2: $N_{\mathrm{hw}} = 20$ held-out pairs executed as 40 independent QPU jobs (256 shots each, \texttt{SamplerV2}).

\subsection{Information-Theoretic Leakage}
Table~\ref{tab:hardware_mi} reports MI, effect size, and Wilcoxon statistics from Stage~1 ($N_{\mathrm{perm}} = 500$). All nine features exceed the null 95\% threshold ($p < 0.0001$). Seven features saturate the binary MI ceiling (0.695 bits), indicating near-perfect information transfer. Ratio-based features show extreme $d_z$ values (up to 39.33).

\begin{table}[ht]
\centering
\caption{Topological fingerprinting on \texttt{ibm\_fez} coupling map (Stage~1, $n=120$ pairs). $\Delta$ = Periodic $-$ Dirichlet. All nine features $p < 0.0001$. $\dagger$: significance driven by distributional shape.}
\label{tab:hardware_mi}
\setlength{\tabcolsep}{4pt}
\footnotesize
\begin{tabular}{lcccr}
\toprule
Feature & MI (bits) & $\Delta$ & $d_z$ & $p$ \\
\midrule
\texttt{swap\_fraction}  & 0.695 & $+0.08$  & $+39.33$ & $<$0.0001 \\
\texttt{cx\_fraction}    & 0.695 & $+0.14$  & $+22.13$ & $<$0.0001 \\
\texttt{twoq\_overhead}  & 0.695 & $+1.12$  & $+14.45$ & $<$0.0001 \\
\texttt{extra\_depth}    & 0.695 & $+38.94$ & $+3.69$  & $<$0.0001 \\
\texttt{depth\_overhead} & 0.695 & $+38.94$ & $+3.69$  & $<$0.0001 \\
\texttt{extra\_twoq}     & 0.695 & $+42.05$ & $+3.40$  & $<$0.0001 \\
\texttt{swap\_equiv}     & 0.695 & $+14.02$ & $+3.40$  & $<$0.0001 \\
\texttt{routed\_depth}   & 0.513 & $+48.22$ & $+3.79$  & $<$0.0001 \\
\texttt{transpile\_ms}   & 0.119 & $-0.10$  & $-0.01$  & $<$0.0001$^\dagger$ \\
\bottomrule
\end{tabular}
\end{table}

\subsection{Classification and Real QPU Execution}
Both Logistic Regression and Random Forest achieve Macro-F1 $= 1.000 \pm 0.000$ on Stage~1 (grouped 5-fold CV), the maximum empirical SCI-IND advantage $\widehat{\mathrm{Adv}} = 0.500$. Table~\ref{tab:hw_qpu} reports artefacts from the 40 real QPU jobs (Stage~2). All nine features remain significant. The empirical routing gap is $\Delta_{\mathrm{routing}} = 32.85$ \texttt{extra\_twoq} gates ($d_z = 4.981$). The headline \texttt{twoq\_overhead} shows $d_z = 7.701$, reflecting near-deterministic girth-forced routing. A provider observing only post-compilation two-qubit gate ratios can distinguish hidden boundary intent with near certainty.

\begin{table*}[ht]
\centering
\caption{Real QPU artefacts on \texttt{ibm\_fez} (Stage~2, $n=20$ matched pairs, 40 QPU jobs).}
\label{tab:hw_qpu}
\setlength{\tabcolsep}{3.5pt}
\footnotesize
\begin{tabular}{lrrrrl}
\toprule
Feature & D mean & P mean & $\Delta$ & $d_z$ & $p$ \\
\midrule
\texttt{twoq\_overhead} & 1.624 & 2.519 & $+0.895$ & $+7.701$ & $<$0.0001 \\
\texttt{swap\_fraction} & 0.127 & 0.201 & $+0.074$ & $+5.224$ & $<$0.0001 \\
\texttt{swap\_equiv}    & 8.000 & 18.95 & $+10.95$ & $+4.981$ & $<$0.0001 \\
\texttt{extra\_twoq}    & 24.00 & 56.85 & $+32.85$ & $+4.981$ & $<$0.0001 \\
\texttt{routed\_depth}  & 145.3 & 247.7 & $+102.4$ & $+3.425$ & $<$0.0001 \\
\texttt{depth\_overhead}& 72.80 & 165.8 & $+93.00$ & $+3.389$ & $<$0.0001 \\
\texttt{extra\_depth}   & 72.80 & 165.8 & $+93.00$ & $+3.389$ & $<$0.0001 \\
\texttt{cx\_fraction}   & 0.209 & 0.247 & $+0.037$ & $+2.924$ & $<$0.0001 \\
\texttt{transpile\_ms}  & 11.51 & 12.50 & $+0.981$ & $+0.571$ & 0.0073   \\
\bottomrule
\end{tabular}
\end{table*}

\section{Artefact-Aware Defence Evaluation}
\label{sec:defence}

We evaluate three gate-padding strategies on \texttt{ibm\_fez} hardware parameters (CZ error 0.003, readout error 0.010):
\textbf{D1}~, two-qubit overhead normalisation (inject $32.85$ dummy CZ pairs to close $\Delta_{\mathrm{routing}}$);
\textbf{D2}~, depth normalisation (equalise \texttt{routed\_depth});
\textbf{D3}~, combined D1 and D2.

SCaaS platforms re-transpile submitted circuits against live device calibration data at execution time, overwriting any client-side routing. The client cannot control the final routing pass; provider-side artefacts are unavoidable. Moreover, classical cloud analogues (SLAs, confidential-computing enclaves) rely on a trusted execution environment that can enforce constant-time compilation. In quantum platforms, routing is performed against live calibration data by the provider, so guaranteeing the absence of topological leakage requires a new technical primitive; legal or contractual measures alone cannot close the channel.

Normalising a single artefact channel achieves zero privacy gain (Macro-F1 $= 1.000$) while fidelity degrades by up to $13.7\%$. The leakage is distributed across a correlated family of channels. D3 at full budget reduces Dirichlet fidelity from $0.666$ to $0.472$ (a 29.2\% absolute degradation) while adversary Macro-F1 remains $1.000$. The routing gap has a \emph{quantitative} component (raw gate-count difference, closeable with dummy gates) and a \emph{structural} component (the SWAP pattern dictated by girth constraints). $C_8$ forces multi-hop resolution; $C_3$ is resolved locally. Adding dummy gates equalises counts but leaves the structural pattern intact, and the classifier exploits ratio features (\texttt{swap\_fraction}, \texttt{twoq\_overhead}) that retain this topological signal. Fig.~\ref{fig:defence_frontier} shows the privacy-fidelity frontier. All strategies trace a horizontal line at Macro-F1 $= 1.000$ as fidelity drops below $0.50$. There is no privacy-fidelity tradeoff in the conventional sense; the undefended circuit is the Pareto-optimal operating point.Gate-count padding is insufficient because the primary leakage signal is topological. Effective SCI-IND-secure execution requires either a lower-girth device that natively embeds both circuit families, or a provider-side topology-blind routing API that applies constant-time compilation regardless of input connectivity. Compilation latency (\texttt{transpile\_ms}) can be masked by uniform scheduling delays, providing a minimum viable baseline for that channel.
\begin{figure}[t]
    \centering
    \includegraphics[width=\columnwidth]{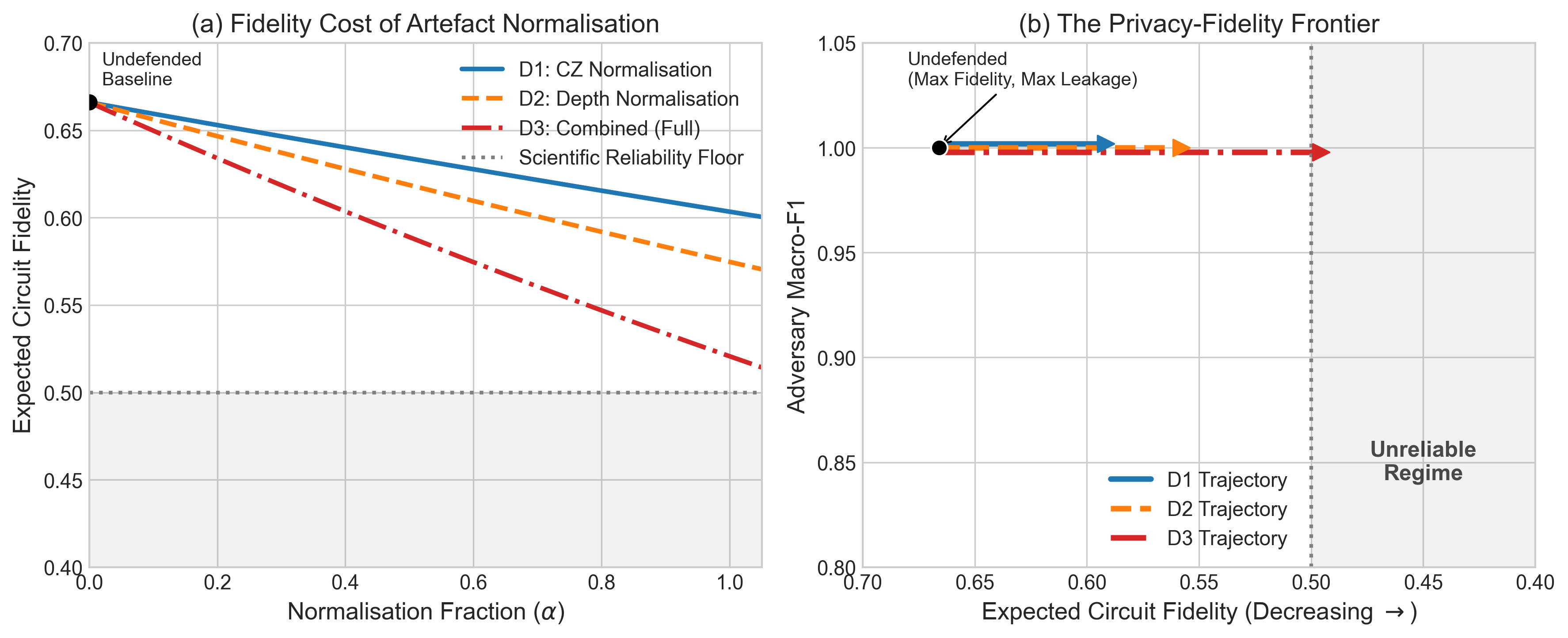}
    \caption{Privacy-fidelity frontier for D1, D2, D3 on \texttt{ibm\_fez}.}
    \label{fig:defence_frontier}
\end{figure}

\section{Conclusion}
\label{sec:conclusion}
We have introduced Scientific-Intent Indistinguishability (SCI-IND), proved that passive SCI-IND security is asymptotically unachievable under routing-optimal compilation on constrained hardware, and validated the resulting leakage on a 156-qubit IBM Heron~r2 QPU ($d_z = 7.701$, Macro-F1 $= 1.000$). These results establish a new confidentiality challenge: protecting quantum states is insufficient when compilation artefacts encode hidden scientific structure. Execution-level confidentiality must become a first-class design requirement, with topology-blind routing APIs, artefact normalisation, and schedule-level timing masking as minimum baselines.

The threat extends to any delegated workload whose hidden intent induces non-isomorphic operator connectivity. Importantly, the present study is limited to fixed-connectivity superconducting hardware; architectures with native reconfigurable or all-to-all connectivity (e.g., neutral atoms with shuttling) may be structurally immune, a direction that warrants future investigation. Formalising execution-trace privacy guarantees across these domains and developing principled low-overhead defences remain open problems.

\bibliographystyle{IEEEtran}
\bibliography{references}

@article{Montanaro_2016,
   title={Quantum algorithms and the finite element method},
   volume={93},
   ISSN={2469-9934},
   url={http://dx.doi.org/10.1103/PhysRevA.93.032324},
   DOI={10.1103/physreva.93.032324},
   number={3},
   journal={Physical Review A},
   publisher={American Physical Society (APS)},
   author={Montanaro, Ashley and Pallister, Sam},
   year={2016},
   month=mar }

@article{Harrow_2009,
   title={Quantum Algorithm for Linear Systems of Equations},
   volume={103},
   ISSN={1079-7114},
   url={http://dx.doi.org/10.1103/PhysRevLett.103.150502},
   DOI={10.1103/physrevlett.103.150502},
   number={15},
   journal={Physical Review Letters},
   publisher={American Physical Society (APS)},
   author={Harrow, Aram W. and Hassidim, Avinatan and Lloyd, Seth},
   year={2009},
   month=oct }

@article{Costa_2019,
   title={Quantum algorithm for simulating the wave equation},
   volume={99},
   ISSN={2469-9934},
   url={http://dx.doi.org/10.1103/PhysRevA.99.012323},
   DOI={10.1103/physreva.99.012323},
   number={1},
   journal={Physical Review A},
   publisher={American Physical Society (APS)},
   author={Costa, Pedro C. S. and Jordan, Stephen and Ostrander, Aaron},
   year={2019},
   month=jan }

@article{Kyriienko_2021,
   title={Solving nonlinear differential equations with differentiable quantum circuits},
   volume={103},
   ISSN={2469-9934},
   url={http://dx.doi.org/10.1103/PhysRevA.103.052416},
   DOI={10.1103/physreva.103.052416},
   number={5},
   journal={Physical Review A},
   publisher={American Physical Society (APS)},
   author={Kyriienko, Oleksandr and Paine, Annie E. and Elfving, Vincent E.},
   year={2021},
   month=may }

@article{Childs_2021,
   title={High-precision quantum algorithms for partial differential equations},
   volume={5},
   ISSN={2521-327X},
   url={http://dx.doi.org/10.22331/q-2021-11-10-574},
   DOI={10.22331/q-2021-11-10-574},
   journal={Quantum},
   publisher={Verein zur Forderung des Open Access Publizierens in den Quantenwissenschaften},
   author={Childs, Andrew M. and Liu, Jin-Peng and Ostrander, Aaron},
   year={2021},
   month=nov, pages={574} }

@misc{lu2024quantumleaktimingsidechannel,
      title={Quantum Leak: Timing Side-Channel Attacks on Cloud-Based Quantum Services}, 
      author={Chao Lu and Esha Telang and Aydin Aysu and Kanad Basu},
      year={2024},
      eprint={2401.01521},
      archivePrefix={arXiv},
      primaryClass={cs.ET},
      url={https://arxiv.org/abs/2401.01521}, 
}

@misc{dong2025exploitingtimingsidechannelsquantum,
      title={Exploiting Timing Side-Channels in Quantum Circuits Simulation Via ML-Based Methods}, 
      author={Ben Dong and Hui Feng and Qian Wang},
      year={2025},
      eprint={2509.12535},
      archivePrefix={arXiv},
      primaryClass={cs.CR},
      url={https://arxiv.org/abs/2509.12535}, 
}

@article{Gheorghiu_2015,
   title={Robustness and device independence of verifiable blind quantum computing},
   volume={17},
   ISSN={1367-2630},
   url={http://dx.doi.org/10.1088/1367-2630/17/8/083040},
   DOI={10.1088/1367-2630/17/8/083040},
   number={8},
   journal={New Journal of Physics},
   publisher={IOP Publishing},
   author={Gheorghiu, Alexandru and Kashefi, Elham and Wallden, Petros},
   year={2015},
   month=aug, pages={083040} }

@inproceedings{Broadbent_2009,
   title={Universal Blind Quantum Computation},
   url={http://dx.doi.org/10.1109/FOCS.2009.36},
   DOI={10.1109/focs.2009.36},
   booktitle={2009 50th Annual IEEE Symposium on Foundations of Computer Science},
   publisher={IEEE},
   author={Broadbent, Anne and Fitzsimons, Joseph and Kashefi, Elham},
   year={2009},
   month=oct, pages={517–526} }

@misc{fitzsimons2016privatequantumcomputationintroduction,
      title={Private quantum computation: An introduction to blind quantum computing and related protocols}, 
      author={Joseph F. Fitzsimons},
      year={2016},
      eprint={1611.10107},
      archivePrefix={arXiv},
      primaryClass={quant-ph},
      url={https://arxiv.org/abs/1611.10107}, 
}

@article{Brearley_2024,
   title={Quantum algorithm for solving the advection equation using Hamiltonian simulation},
   volume={110},
   ISSN={2469-9934},
   url={http://dx.doi.org/10.1103/PhysRevA.110.012430},
   DOI={10.1103/physreva.110.012430},
   number={1},
   journal={Physical Review A},
   publisher={American Physical Society (APS)},
   author={Brearley, Peter and Laizet, Sylvain},
   year={2024},
   month=jul }

@inproceedings{tannu2019mitigating,
author = {Tannu, Swamit S. and Qureshi, Moinuddin K.},
title = {Not All Qubits Are Created Equal: A Case for Variability-Aware Policies for NISQ-Era Quantum Computers},
year = {2019},
isbn = {9781450362405},
publisher = {Association for Computing Machinery},
address = {New York, NY, USA},
url = {https://doi.org/10.1145/3297858.3304007},
doi = {10.1145/3297858.3304007},
booktitle = {Proceedings of the Twenty-Fourth International Conference on Architectural Support for Programming Languages and Operating Systems},
pages = {987–999},
numpages = {13},
location = {Providence, RI, USA},
series = {ASPLOS '19}
}

@misc{dasgupta2020characterizingstabilitynisqdevices,
      title={Characterizing the Stability of NISQ Devices}, 
      author={Samudra Dasgupta and Travis S. Humble},
      year={2020},
      eprint={2008.09612},
      archivePrefix={arXiv},
      primaryClass={quant-ph},
      url={https://arxiv.org/abs/2008.09612}, 
}

@inproceedings{das2019case,
  author    = {Das, Poulami and Tannu, Swamit S. and Nair, Prashant J. and Qureshi, Moinuddin},
  title     = {A Case for Multi-Programming Quantum Computers},
  booktitle = {Proceedings of the 52nd Annual {IEEE/ACM} International Symposium on Microarchitecture},
  series    = {MICRO-52},
  year      = {2019},
  pages     = {291--303},
  location  = {Columbus, OH, USA},
  doi       = {10.1145/3352460.3358287}
}

@inproceedings{saki2020analysis,
  author    = {Saki, Abdullah Ash and Alam, Mahabubul and Ghosh, Swaroop},
  title     = {Analysis of crosstalk in {NISQ} devices and security implications in multi-programming regime},
  booktitle = {Proceedings of the 39th International Conference on Computer-Aided Design}, 
  year      = {2020},
  month     = {8},
  pages     = {25--30},
  doi       = {10.1145/3370748.3406570}
}
\end{document}